\documentclass{article}
\usepackage[english]{babel}
\usepackage[utf8]{inputenc}
\usepackage[T1]{fontenc}
\usepackage[a4paper, margin=1in]{geometry}
\usepackage{graphicx}
\graphicspath{{figures/}}
\usepackage{framed}
\usepackage[framemethod=tikz]{mdframed}
\usepackage{todonotes}
\usepackage{color}

\definecolor{darkgreen}{rgb}{0,0.5,0}
\usepackage{hyperref}
\hypersetup{
    unicode=false,          
    colorlinks=true,        
    linkcolor=blue,          
    citecolor=darkgreen,        
    filecolor=magenta,      
    urlcolor=cyan           
}

\usepackage{amsthm}
\usepackage{amsmath}
\usepackage{amssymb}
\usepackage{amsfonts}
\usepackage{mathrsfs}
\usepackage{mathtools}
\usepackage{verbatim}
\usepackage{footnote}
\usepackage{algorithm}
\usepackage{algpseudocode}
\usepackage{lineno}
\usepackage{caption}
\usepackage{framed}
\usepackage{enumerate}

\usepackage{tikzsymbols}
\usepackage{thmtools,thm-restate}
\usepackage{nicefrac}

\usepackage{subcaption}
\usepackage{makecell}

\usepackage[capitalize, nameinlink]{cleveref}
\usepackage[section]{placeins}
\crefname{theorem}{Theorem}{Theorems}
\Crefname{lemma}{Lemma}{Lemmas}
\Crefname{invariant}{Invariant}{Invariants}
\Crefname{claim}{Claim}{Claims}
\Crefname{observation}{Observation}{Observations}
\Crefname{@algorithm}{Algorithm}{Algorithms}
\Crefname{figure}{Figure}{Figures}

\newtheorem{theorem}{Theorem}[section]
\newtheorem{lemma}[theorem]{Lemma}

\newtheorem*{remark*}{Remark}
\newtheorem{result}{Result}

\newcommand{\eps}{\varepsilon}
\newcommand{\eqdef}{\stackrel{\text{\tiny\rm def}}{=}}

\newcommand{\E}[1]{\mathbb{E}\left[#1\right]}
\newcommand{\prob}[1]{\Pr \left[ #1 \right]}

\newcommand{\tS}{\tilde{S}}
\newcommand{\tT}{\tilde{T}}

\newcommand{\cAMPC}{\mathcal{A}_{\textsc{MPC}}}

\newcommand{\cM}{\mathcal{M}}

\newcommand{\VSetsUpdate}{\textsc{VSets-Update}\xspace}
\newcommand{\SSSample}{\textsc{Set-Sample}\xspace}
\newcommand{\din}{d^-}
\newcommand{\dout}{d^+}
\newcommand{\eqd}{\stackrel{\mathclap{\tiny\mbox{def}}}{=}}
\newcommand{\rb}[1]{\left( #1 \right)}
\newcommand{\cb}[1]{\left\{ #1 \right\}}
\newcommand{\Estream}{E_{\rm{stream}}}
\newcommand{\Erel}{E_{\rm{rel}}}
\newcommand{\Erem}{E_{\rm{rem}}}

\newcommand{\copt}{c_{\rm OPT}}
\newcommand{\dinrem}{d_{\rm in}^{\rm rem}}
\newcommand{\doutrem}{d_{\rm out}^{\rm rem}}

\DeclareMathOperator{\poly}{poly}

\begin{document}

\newcommand\relatedversion{}
\renewcommand\relatedversion{\thanks{The full version of the paper can be accessed at \protect\url{https://arxiv.org/abs/1902.09310}}} 

\title{Faster Streaming and Scalable Algorithms for \\
Finding Directed Dense Subgraphs in Large Graphs}
\author{Slobodan Mitrovi\'{c}\thanks{UC Davis, smitrovic@ucdavis.edu}
\and Theodore Pan\thanks{UC Davis, thjpan@ucdavis.edu}}

\date{}

\maketitle

\pagenumbering{arabic}

\begin{abstract}
Finding dense subgraphs is a fundamental algorithmic tool in data mining, community detection, and clustering.
In this problem, one aims to find an induced subgraph whose edge-to-vertex ratio is maximized.

We study the directed case of this question in the context of semi-streaming and massively parallel algorithms. 
In particular, we show that it is possible to find a $(2+\epsilon)$ approximation on randomized streams even in a single pass by using $O(n \cdot {\rm poly} \log n)$ memory on $n$-vertex graphs. 
Our result improves over prior works, which were designed for arbitrary-ordered streams: the algorithm by Bahmani et al.~(VLDB 2012) which uses $O(\log n)$ passes, and the work by Esfandiari et al.~(2015) which makes one pass but uses $O(n^{3/2})$ memory. 
Moreover, our techniques extend to the Massively Parallel Computation model yielding $O(1)$ rounds in the super-linear and $O(\sqrt{\log n})$ rounds in the nearly-linear memory regime. 
This constitutes a quadratic improvement over state-of-the-art bounds by Bahmani et al.~(VLDB 2012 and WAW 2014), which require $O(\log n)$ rounds even in the super-linear memory regime.

Finally, we empirically evaluate our single-pass semi-streaming algorithm on $6$ benchmarks and show that, even on non-randomly ordered streams, the quality of its output is essentially the same as that of Bahmani et al.~(VLDB 2012) while it is $2$ times faster on large graphs.
\end{abstract}

\section{Introduction}
Given a directed graph $G = (V, E)$, the directed \emph{densest subgraph problem} asks to find two vertex subsets $S, T \subseteq V$ such that the number of edges from $S$ to $T$ scaled by $\sqrt{|S| \cdot |T|}$ is maximized. 
When $S = T$, this problem is equivalent to finding the densest subgraph in undirected graphs.
Dense subgraph discovery is a fundamental algorithmic tool in data mining~\cite{kriegel2005density,wu2019density,fang2022densest}, community detection~\cite{chen2010dense,harenberg2014community}, spam detection~\cite{leon2011web,zhang2016detecting}, fraud discovery~\cite{zhang2017hidden,ren2021ensemfdet}, clustering~\cite{kriegel2011density,bhattacharjee2021survey}, graph compression~\cite{buehrer2008scalable}, and many other applications.
Given its importance, this problem has been studied in various computational settings, including semi-streaming~\cite{bahmani2012densest,mcgregor2015densest,bhattacharya2015space}, dynamic~\cite{epasto2015efficient,sawlani2020near}, distributed~\cite{su2019distributed}, and parallel~\cite{bahmani2012densest,ghaffari2019improved}, with the first study of its undirected version dating back to the eighties~\cite{goldberg1984finding}.

A directed densest subgraph can be found in polynomial time~\cite{charikar2003greedy}. However, the corresponding algorithm solves a family of $O(|V|^2)$ linear programs, making it impractical for execution on large graphs. Nevertheless, much faster approximation algorithms have been developed. For instance, there exists a simple greedy algorithm running in near-linear time that yields a $2 + \eps$ approximation for any arbitrary constant $\eps > 0$~\cite{bahmani2012densest,charikar2003greedy}.
The prevalence of large graphs has motivated researchers to study the approximate densest subgraph problem in settings for processing big data, such as \emph{streaming} and \emph{massively parallel computation (MPC)}.

There has been a significant interest in designing semi-streaming algorithms for finding dense subgraphs, e.g., \cite{bahmani2012densest,esfandiari2015applications,bhattacharya2015space,mcgregor2015densest} and references therein.
Esfandiari et al.~\cite{esfandiari2015applications} designed a single-pass semi-streaming algorithm for finding $(1 + \eps)$-approximate densest subgraph in \emph{undirected} graphs.
For $n$-vertex graph $G$, this algorithm uniformly at random samples $O(n \cdot \log(n) / \eps^2)$ edges from $G$ in the streaming fashion. 
The authors show that the densest subgraph in such a sample is a $(1 + \eps)$-approximate densest subgraph in $G$.
This sampling idea fails in the context of directed graphs. 
To see that, consider the example in \cref{fig:difficult-example}.
\begin{figure}[h]
    \centering
    \includegraphics[width=1\linewidth]{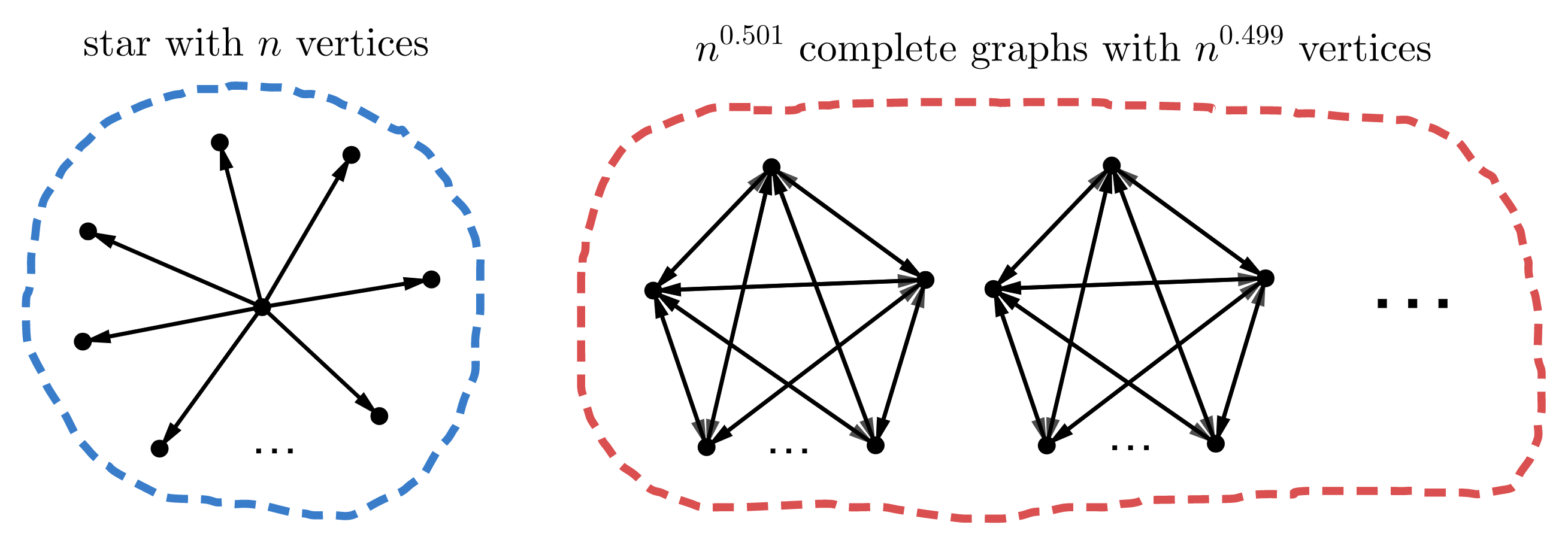}
    \caption{\small An illustration of why uniformly at random sampling applicable to undirected fails in the context of directed densest subgraph.}
    \label{fig:difficult-example}
\end{figure}
In that example, the directed densest subgraph is the star, having density $\Theta(\sqrt{n})$; each clique has density $\Theta(n^{0.499})$. However, sampling $O(n \cdot \log n)$ edges from that graph will, in expectation, contain only $O(n^{0.501})$ edges from the star, which is insufficient to recover the densest subgraph in \cref{fig:difficult-example}. 
To alleviate this, \cite{epasto2015efficient} propose a single-pass semi-streaming algorithm that requires $O(n^{1.5} \poly \log n)$ memory. 
In another line of work, Bahmani et al.~\cite{bahmani2012densest}, the memory requirement is kept at $O(n \poly \log n)$ at the expense of making $O(\log n)$ passes.
It is unknown whether these techniques can be extended to achieve the best of both worlds, i.e., a single pass and $O(n \poly \log n)$ memory.

Finding dense subgraphs in MPC has also received significant attention. 
For undirected graphs, the sampling strategy devised in \cite{esfandiari2015applications} for the streaming setting readily transfers to an $O(1)$ MPC round algorithm in the near-linear memory regime.
For the sublinear memory regime, Bahmani et al.~\cite{bahmani2014efficient} developed an algorithm that in $O(\log(n) / \eps^2)$ rounds outputs a $(1+\eps)$ approximation.
This was improved to $O(\sqrt{\log n} \cdot \log \log n)$ rounds by Ghaffari et al.~\cite{ghaffari2019improved}.
However, state-of-the-art algorithms for densest subgraphs have significantly higher round complexity for directed than for undirected graphs.
The works \cite{bahmani2012densest} and \cite{bahmani2014efficient} also design an algorithm for finding $(2+\eps$)- and $(1+\eps)$-approximate densest subgraphs in directed graphs in $O(\log(n) / \eps)$ and $O(\log(n) / \eps^2)$ MPC rounds, respectively. 
To the best of our knowledge, this is the most efficient MPC approach, even for the super-linear memory regime, leaving a considerable gap between the case of directed and undirected graphs: $O(\log n)$ vs.~$O(1)$ round complexity.
Given that many modern graphs are directed -- examples include Twitter and Instagram '' following '' relationships, networks of sent emails, and even links on the Web -- \emph{it is natural to wonder whether faster algorithms for finding directed densest subgraphs can be developed}.

\subsection{Our contributions}
As our first result, we make progress in the semi-streaming setting.
\begin{result}[\cref{theorem:semi-streaming} rephrased]
\label{result:semi-streaming}
    Given an $n$-vertex graph, there exists a single-pass semi-streaming algorithm that, with high probability, over randomized streams outputs a $(2+\eps)$-approximate directed densest subgraph while using $O(n \poly(\log(n) / \eps))$ memory.
\end{result}
We build on ideas we develop for \cref{result:semi-streaming} and improve the state-of-the-art in the context of MPC as well.
\begin{result}[\cref{theorem:MPC-superlinear,theorem:MPC-nearlylinear} summarized]
    Given an $n$-vertex graph, there exists an MPC algorithm that, with high probability, outputs a $(2+\eps)$-approximate directed densest and
    \begin{itemize}
        \item when the memory per machine is $n^{1+\delta}$ for an arbitrary constant $\delta > 0$, the algorithm runs in $O(1)$ rounds;
        \item when the memory per machine is $n \poly(\log(n) / \eps)$, the algorithm runs in $O(\sqrt{\log n})$ rounds.
    \end{itemize}
\end{result}
The fastest previously known MPC algorithms require $O(\log n)$ rounds even when the memory per machine is $n^{1+\delta}$~\cite{bahmani2012densest,bahmani2014efficient}. So, our approach makes significant progress toward closing the gap between the undirected and directed cases.

Moreover, our empirical evaluations show that our semi-streaming algorithm is at least $2$ times faster than Bahmani et al.~\cite{bahmani2012densest} while achieving essentially the same accuracy. It is interesting to note that our algorithm is by several percent more accurate than \cite{bahmani2012densest} on large graphs.
These hold even when the underlying stream is non-randomized.

\section{Preliminaries}
\paragraph{Notation.}
Given a directed graph $G = (V, E)$, we use $E_G$ to denote $E$, the set of edges specific to graph $G$. We use $n$ to refer to $|V|$.
Given two vertex sets $S, T\subseteq V$, we refer to the edges between them by $E_G(S, T) \eqd \{e = (i, j) \in E_G : i\in S, j\in T\}$.
For the sake of brevity, for $i, j \in V$, we also write $E_G(i, T) \eqd E_G(\{i\}, T)$ and $E_G(S, j) \eqd E_G(S, \{j\})$.
Specific to streams, we use $\Estream(S,T)$ to denote the subset of $E_G(S,T)$ remaining in the stream.
For vertex $v \in V$, we use $\din_G(v)$ to denote the in-degree and $\dout_G(v)$ to denote the out-degree of $v$ in $G$. 

When we say that an event $A$ happens \emph{with high probability}, or whp for short, we imply that $\prob{A} \ge 1 - n^{-c}$, for a constant $c>0$. In our algorithms, the constant $c$ can be made arbitrarily large by paying constant factors in the pass/round or memory complexity.

\paragraph{Directed densest subgraph.} Given directed graph $G = (V, E)$ and vertex sets $S,T\subseteq V$, the \textit{density} $\rho(S, T)$ is defined as $\rho(S, T) \eqd |E_G(S,T)|/\sqrt{|S| \cdot |T|}$. A \textit{densest subgraph} is sets $S^*$, $T^*$ such that $(S^*, T^*) \in \arg \max_{S,T\subseteq V} \rho(S, T)$.

\paragraph{Semi-streaming.}
In the semi-streaming model, an algorithm scans input data, e.g., scans edge by edge of an input graph. After all the edges are scanned, we say that the algorithm made a \emph{pass}.
When the edges of the graph are presented to the algorithm in a random permutation, we say that the stream is \emph{randomized}.
Throughout the process, for $n$-vertex graph, the algorithm is typically allowed to use $O(n \poly \log n)$ memory. 
The complexity measure in this setup is the number of passes the algorithm makes and the memory it requires; as discussed for \cite{esfandiari2015applications}, some approaches require polynomially more than memory $n$. 
After all the passes are performed, the algorithm outputs a solution.

\paragraph{Massively Parallel Computation (MPC).}
The Massively Parallel Computation (MPC) model has become a standard in the theoretical study of practical large-scale frameworks~\cite{karloff2010model,goodrich2011sorting,beame2017communication} such as MapReduce, Hadoop, Spark, and Flume.
In MPC, the computation is performed in synchronous rounds across $N$ machines. Each machine has $S$ words of memory.
Originally, input data is partitioned arbitrarily across the machines. During a round, each machine performs computation on its data locally. At the end of a round, machines exchange messages. The communication topology among the machines is a clique. During a round, a machine can, in total, send (receive) at most $S$ bits of data.
There are three important memory regimes with respect to $S$ relative to data size. Let the input be an $n$-vertex graph, and let $\delta \in (0, 1)$ be an arbitrary constant: \emph{sub-linear} corresponds to $S = n^{\delta}$; \emph{near-linear} corresponds to $S = n \poly \log n$; and, $S = n^{1 + \delta}$ corresponds to \emph{super-linear}.

\paragraph{Probability tools.} In our analysis, we extensively apply the following well-known tool from probability.
\begin{theorem}[Chernoff bound]\label{lemma:chernoff}
	Let $X_1, \ldots, X_k$ be independent random variables taking values in $[0, 1]$. Let $X \eqdef \sum_{i = 1}^k X_i$ and $\mu \eqdef \E{X}$. Then,
	\begin{enumerate}[(A)]
		\item\label{item:delta-at-most-1} For any $\delta \in [0, 1]$ it holds $\prob{|X - \mu| \ge \delta \mu} \le 2 \exp\rb{- \delta^2 \mu / 3}$.
		\item\label{item:delta-at-most-1-ge} For any $\delta \in [0, 1]$ it holds $\prob{X \ge (1 + \delta) \mu} \le \exp\rb{- \delta^2 \mu / 3}$.
	\end{enumerate}
\end{theorem}
\section{The Base Algorithm}
Let $c$ be a value such that there exists a densest subgraph $(S, T)$ in $G$ and $c = |S| / |T|$.
The starting point of our approach is a classical peeling algorithm for the directed densest subgraph problem.
This method receives $c$ as input and iteratively ``peels'', i.e., removes, vertices that are below a certain degree threshold.
Intuitively, the peeling is performed so that the removed vertices do not significantly affect the densest subgraph.
We make this statement formal by \cref{theorem:approx}.
One iteration of this peeling primitive is given by \VSetsUpdate (\cref{alg:peeling-iteration}).


\begin{algorithm}
\caption{(\VSetsUpdate)}\label{alg:peeling-iteration} 

\textbf{Input:} $G = (V, E)$, $c > 0$, $\epsilon \in (0,1)$, and vertex sets $S$, $T$
\begin{algorithmic}[1]
\If{$|S|/|T| \geq c$}

    \State $A(S) \gets \{i\in S : |E(i, T)| \leq (1+\epsilon)\frac{|E(S,T)|}{|S|}\}$ \label{line:peel-S}

    \State $S \gets S\backslash A(S)$
\Else

    \State $B(T) \gets \{j\in T : |E(S, j)| \leq (1+\epsilon)\frac{|E(S,T)|}{|T|}\}$ \label{line:peel-T}

    \State $T \gets T\backslash B(T)$
\EndIf


\State\Return{$(S, T)$}

\end{algorithmic}
\end{algorithm}

By iteratively peeling and keeping track of the vertex sets that produce the densest subgraph, we are able to find an approximation of the directed densest subgraph. \VSetsUpdate has two handy features. The first one is that, given two sets $S$ and $T$, it reduces the size of one of them by $1 + \eps$.
\begin{lemma}
\label{passes}
Let $G = (V, E)$ be a directed graph. For $i \ge 0$, let $(S_i, T_i)$ be a sequence of set-pairs such that $(S_0, V_0) = (V, V)$ and $(S_{i+1}, T_{i+1}) = \VSetsUpdate(G, c, \epsilon, S_i, T_i)$. Then $S_t = \emptyset$ or $T_t = \emptyset$ for some $t\in O(\log(n) / \epsilon)$.
\end{lemma}
\begin{proof}
Consider \VSetsUpdate invoked with $|S_i|/|T_i| \geq c$, so vertices are removed from $S_i$.
Then, all vertices $u\in S_i\backslash A(S_i)$ satisfy $|E(u, T)| > (1+\epsilon)\frac{|E(S_i,T_i)|}{|S_i|}$.
Since $S_{i+1} = S_i\backslash A(S_i)$, we have that $|S_{i+1}| < |S_i|/(1+\epsilon)$.
Similarly, if $|S_i|/|T_i| < c$, we have that $|T_{i+1}| < |T_i| / (1+\epsilon)$.
So, the size of one of the sets decreases by at least $1+\epsilon$ times each time \VSetsUpdate is invoked. Then, one of the sets becomes empty after $O(\log_{1+\epsilon}n) = O(\log(n) / \eps)$ invocations of \VSetsUpdate.
\end{proof}
The second useful feature of \VSetsUpdate is that one of the pairs $(S_i, T_i)$, as defined in the statement of \cref{passes}, is a $(2+\epsilon)$-approximate densest subgraph. We prove such a statement in \cref{theorem:approx}.
\section{$(2+\epsilon)$ Approximation in a Single Pass}
\label{sec:single-pass}
\subsection{A Multi-pass Sampling-based Algorithm}
We modify the base algorithm to sample edges to determine which vertices to peel, instead of using the entire graph for that. By sampling a sufficient number of edges, the degrees of vertices in the sample are used to estimate the degrees of vertices in the original graph.
For each pass, \cref{alg:first} creates a new sample of our graph to apply \VSetsUpdate to.
In our analysis, we extensively use the following threshold
\begin{equation}\label{eq:threshold}
    \xi \eqdef \frac{60\log n}{\epsilon^2}.
\end{equation}
\begin{algorithm}
\caption{A semi-streaming directed densest subgraph algorithm}\label{alg:first}
\textbf{Input:} $G=(V,E)$, $c > 0$, and $\epsilon \in (0,1)$\\
\textbf{Output:} A $(2+\eps)$-approximate directed densest subgraph
\begin{algorithmic}[1]
\State Initialize each of $S, T,S^*,T^*$ to $V$

\While{$S \neq \emptyset$ and $T \neq \emptyset$}

    \State Let $H$ be a sample of $E_G(S, T)$ with each edge sampled independently with probability $\min\cb{\frac{n \xi}{(1-\epsilon)|E_G(S, T)|},1}$ \label{line:multi-pass-define-p}

    \State $(S,T) \gets \VSetsUpdate(H, c, \epsilon, S, T)$

    \If{$\rho(S,T) > \rho(S^*,T^*)$}
        \State $S^*\gets S, T^*\leftarrow T$
    \EndIf

\EndWhile

\State \Return{$S^*,T^*$}
\end{algorithmic}
\end{algorithm}

Before we begin analyzing \cref{alg:first}, we prove \cref{gtoh} as a tool that connects the degrees of vertices of graph $G$ and sampled subgraph $H$. 
\begin{lemma}
\label{gtoh}
Let $G = (V, E)$ be a graph, $\epsilon \in (0,1)$, and $p\in (0,1]$. Let $H$ be a subgraph of $G$ that contains each edge of $G$ independently with probability $p$.
Then, for all $v \in G$, the following hold with probability $1 - \frac{1}{n^3}$: (i) if $\din_G(v) \geq \xi / p$ then $p\din_G(v) \leq (1+\epsilon)\din_H(v)$, (ii) if $\din_G(v) < \xi / p$ then $\din_H(v) < 2 \xi$ (same claims hold for out-degrees $\dout_G(v)$ and $\dout_H(v)$).
\end{lemma}
\begin{proof}
Consider a vertex $v\in V$. Let $X$ be the random variable equal to $\din_H(v)$. Observe that $\E{X} = p \din_G(v)$. Then, if $\din_G(v) \geq \xi/p$, we have
\begin{eqnarray}
 & & \Pr(\E{X} \geq (1+\epsilon)X)\nonumber\\
&\leq& \Pr\left(\E{X} - X \geq \frac{\epsilon}{2} \E{X}\right)\nonumber\\
&\leq& 3\exp\left(-\epsilon^2 \E{X} / 12\right)\label{eq:gtoh2}\\
&=& 3\exp\left(-\epsilon^2 p\din_G(v) / 12\right)\nonumber\\
&\leq& 3\exp(-5\log (n))\nonumber\\
&=& \frac{3}{n^5} \leq \frac{1}{n^4}\nonumber
\end{eqnarray}
assuming $n \geq 3$, where we used the Chernoff bound, \cref{lemma:chernoff}~\eqref{item:delta-at-most-1}, to derive \eqref{eq:gtoh2}.

Now, if $\din_G(v) < \xi/p$, we have
\begin{eqnarray}
& & \Pr\left(X \geq 2\xi\right)\nonumber\\
&\leq& \Pr\left(X - \E{X} \geq \xi \right)\nonumber\\
&\leq& \frac{3}{2}\exp(-\xi^2/(3 \E{X}))\label{eq:gtoh3}\\
&\leq& \frac{3}{2}\exp(-\xi/3)\nonumber\\
&\leq& \frac{1}{n^4}\nonumber
\end{eqnarray}
where \eqref{eq:gtoh3} follows by the Chernoff bound.

Therefore, using the union bound, the claim holds with probability $1 - \frac{1}{n^3}$. The same proof applies to out-degrees.
\end{proof}

Now, using \cref{gtoh}, we begin analyzing \cref{alg:first}.

\begin{lemma}
\label{firstmem}
\cref{alg:first} uses $O\left(\frac{n\log^2(n)}{\epsilon^3}\right)$ memory with probability $1 - \frac{\log_{1+\epsilon}n}{n^4}$.
\end{lemma}
\begin{proof}
By the Chernoff bound, \cref{lemma:chernoff}~\eqref{item:delta-at-most-1-ge}, every pass of \cref{alg:first} samples $O\left(\frac{(1+\epsilon)n \xi}{(1-\epsilon)}\right) = O\rb{\frac{n\log n}{\epsilon^2}}$ edges with probability $1 - \frac{1}{n^4}$. 
From \cref{passes}, the algorithm makes $O(\log_{1+\epsilon} n)$ passes. Therefore, by the union bound over all the passes, the total memory used by \cref{alg:first} is $O\left(n \xi \cdot \log_{1+\epsilon} n\right) = O\rb{\frac{n\log^2 n}{\epsilon^3}}$ with probability $1 - \frac{\log_{1+\epsilon}n}{n^4}$.
\end{proof}
\begin{theorem}\label{theorem:approx}
\cref{alg:first} produces a $2\left(1 +\epsilon\right)$-approximation of the densest directed subgraph with probability $1 - \frac{\log_{1+\epsilon}n}{n^2}$.
\end{theorem}
\begin{proof}
Consider a pass with $|S|/|T|\geq c$. Hence, vertices are removed from $S$ in that pass.
Let $p = n \xi / ((1-\epsilon)|E_G(S, T)|)$; the same probability is used on \cref{line:multi-pass-define-p} of \cref{alg:first}.
Then, for all vertices $i\in A(S)$, where $A(S)$ is defined in \VSetsUpdate, we have
\begin{eqnarray}
\sqrt{c}|E_G(i,T)| & \leq& \sqrt{c}\cdot\frac{(1+\epsilon)^2|E_H(S,T)|}{p|S|} \label{eq:second-bound}\\
&\leq& \sqrt{c}\cdot\left(1+\epsilon\right)^3\frac{|E_G(S,T)|}{|S|}\label{eq:third-bound}\\\
&\leq&  \sqrt{c}\cdot\left(1 +\epsilon\right)^3\frac{|E_G(S,T)|}{\sqrt{c|S||T|}} \nonumber\\
&=& \left(1 +\epsilon\right)^3\frac{|E_G(S,T)|}{\sqrt{|S||T|}}\nonumber\\
&=& \left(1 +\epsilon\right)^3\rho(S,T).\label{eq:E_G(i,T)-upper-bound}
\end{eqnarray}
For \eqref{eq:second-bound}, we either have that $|E_G(i,T)| \geq \xi / p$ or $|E_G(i,T)| < \xi / p$. If $|E_G(i,T)| \geq \xi / p$, we use \cref{gtoh} to get 
\begin{eqnarray*}
\sqrt{c}|E_G(i,T)| &\leq& \sqrt{c}\cdot \frac{(1+\epsilon)|E_H(i,T)|}{p}\\
&\leq& \sqrt{c}\cdot\frac{(1+\epsilon)^2|E_H(S,T)|}{p|S|}.
\end{eqnarray*}
On the other hand, if $|E_G(i,T)| < \xi / p$, \eqref{eq:second-bound} follows directly.
Additionally, since $\E{|E_H(S, T)|} = p |E_G(S, T)|$, a direct application of Chernoff bound, \cref{lemma:chernoff}~\eqref{item:delta-at-most-1-ge}, yields \eqref{eq:third-bound}. 

Similarly, for a pass with $|S|/|T|<c$, we have that for all vertices $j\in B(T)$,
\begin{equation}\label{eq:E_G(S,j)-upper-bound}
    \frac{1}{\sqrt{c}}|E_G(S,j)| \leq \left(1 +\epsilon\right)^3\rho(S,T).
\end{equation}
We define $\doutrem(i) \eqdef |E_G(i, T)|$ and $\dinrem(j) \eqdef |E_G(S,j)|$ where $T$ is the respective set during the pass $i$ is removed from $S$ and, similarly, $S$ is the respective set during the pass $j$ is removed from $T$.
Let $d_{out}^* = \max_{i\in V}\{\doutrem(i)\}$ and $d_{in}^* = \max_{j\in V}\{\dinrem(j)\}$.
Then, for $\tS,\tT \subseteq V$ with $|\tS|/|\tT| = c$ that maximizes $\rho(\tS,\tT) = \rho^*(G)$, we have that
\begin{eqnarray}
\rho^*(G) &=& \frac{|E(\tS,\tT)|}{\sqrt{|\tS||\tT|}} \leq \frac{|\tS|\cdot d_{out}^* + |\tT|\cdot d_{in}^*}{\sqrt{|\tS||\tT|}}\nonumber\\
&=& \sqrt{c}\cdot d_{out}^* + \frac{1}{\sqrt{c}}\cdot d_{in}^*\leq 2\left(1 +\epsilon\right)^3\rho(S^*,T^*)\label{eq:fourth-bound}
\end{eqnarray}
where \eqref{eq:fourth-bound} comes from using \eqref{eq:E_G(i,T)-upper-bound} and \eqref{eq:E_G(S,j)-upper-bound}.
(The last chain of inequalities is also shown in \cite{charikar2003greedy}.)

Therefore, the algorithm produces a $2\left(1 +\epsilon\right)^3$-approximation. Note that we use \cref{gtoh} on each vertex for every pass. Therefore, by the union bound, we have a success probability of $1 - n\cdot\log_{1+\epsilon}n\cdot\frac{1}{n^3} = 1 - \frac{\log_{1+\epsilon}n}{n^2}$.
\end{proof}

\subsection{Improved Single-pass Algorithm}
Building on \cref{alg:first}, we develop a single-pass semi-streaming algorithm for $(2+\eps)$-approximate directed densest subgraph when the stream of edges is randomized.
The main idea behind our algorithm is to leverage the additional randomization of the stream and avoid performing a new pass to estimate vertex degrees when $S$ or $T$ gets updated; these degrees are needed to execute \cref{line:peel-S,line:peel-T} of \VSetsUpdate.

Let $(S', T') = \VSetsUpdate(H, c, \eps, S, T)$. Observe that $|E(u, T')|$ and $|E(S', v)|$, for $u \in S'$ and $v \in T'$, can be estimated by uniformly at random sampling $O(n \xi)$ edges from $E(S', T')$.
\emph{How can those edges be sampled without performing a new pass?}
If our algorithm has not already seen any edge in $E(S', T')$ from the stream, this task would be easy -- collect edges in $E(S', T')$ from the stream as long as fewer than the desired many are seen.
However, our algorithm already has a subset $E'$ of $E(S', T')$ in its memory, which should be considered. Assuming that the sampling probability is known -- which is not known and we discuss it in a moment -- then \SSSample (\cref{alg:sample}) outputs the desired sample while considering $E'$.

\begin{algorithm}
\caption{(\SSSample)} \label{alg:sample}
\textbf{Input:} A set of edges $E' \subseteq E_G(S,T)$, vertex sets $S$ and $T$, probability $p$, an estimate $s$ of $|E_G(S,T)|$
\textbf{Output:} Uniformly random sample of edges from the union of the stream and $E'$
\begin{algorithmic}[1]
\State Let $H_1$ be a subset of $E'$, where each in $E'$ is included in $H_1$ with probability $p$.

\State Sample $x$ from the binomial distr.~$B(s - |E'|, p)$. \label{line:sample-x}

\State Set $H_2$ to next $x$ edges in $E_G(S,T)$ from the stream.

\State \Return{$H_1\cup H_2$}
\end{algorithmic}
\end{algorithm}
\SSSample as input takes value $p$ where, ideally, $p = n \xi / |E(S, T)|$. Unfortunately, $|E(S, T)|$ is not known; \SSSample has access to $E' \subseteq E(S, T)$ but $\Estream(S, T) = E(S, T) \setminus E'$ is in the rest of the stream still unseen by the algorithm. To alleviate this, our algorithm \emph{approximates} $|E(S, T)|$.
By using \SSSample, this enables generating samples of the $E(S, T)$ sequence without going through the entire stream each time. In doing so, we are able to create \cref{alg:random}, which uses only a single pass.



\begin{algorithm}
\caption{A single-pass randomized-stream directed densest subgraph algorithm}\label{alg:random}
\textbf{Input:} $c > 0$, and $\epsilon \in (0,1)$\\
\textbf{Output:} A $(2+\eps)$-approximate directed densest subgraph
\begin{algorithmic}[1]
\State Initialize each of $S, T,S^*,T^*$ to $V$

\State $E' \gets \emptyset$\Comment{The subset of $E_G(S,T)$ seen so far}

\While{$\Estream \neq \emptyset$}

    \State $E_A \gets$ next $n \xi$ edges from $\Estream$
    \If{$|E_{A}(S,T)| < 2 \xi$ or $\Estream = \emptyset$}
        \State $E'\gets E' \cup E_A(S,T) \cup \Estream(S,T)$
        
        \State If the densest subgraph in $E'$ is denser than $(S^*, T^*)$, update $(S^*, T^*)$ to be that subgraph.

    \Else
        \Comment{Estimate of $|E_G(S,T)|$ based on $E_A$.}
        \State$s\gets (1-\epsilon) \cdot \frac{|E_A(S,T)|}{|E_A|} \cdot \rb{|\Estream| + n\xi} + |E'|$ \label{line:approx-s}
        
        \State $E'\gets E' \cup E_A(S,T)$

        \State $p \gets \frac{n \xi}{(1-\epsilon) s}$

        \State $H \gets \SSSample\rb{E', S, T, p, s}$ \label{line:streaming-H}
        
        \If{$p > 1$}
            \State $H \gets E' \cup \Estream(S, T)$
        \EndIf
        
        \State $(S,T)\gets \VSetsUpdate(H, c, \epsilon, S, T)$
    
        \If{$\rho(S,T) > \rho(S^*,T^*)$}
            \State $S^*\gets S, T^*\leftarrow T$
        \EndIf
        
        \State $E' \leftarrow (E'\cup E_H) \cap E_G(S,T)$
        
        
        
            
    \EndIf

\EndWhile

\State \Return{$S^*,T^*$}
\end{algorithmic}
\end{algorithm}

Notice that $E_A$ in \cref{alg:random} is a sample of $n \xi$ edges taken uniformly at random from \emph{the rest} of the stream.
Additionally, graph $H$ uses \SSSample to create an approximately uniform sample of edges from graph $E(S, T)$. In both cases, the edges of the samples are not chosen independently, so the standard Chernoff bound cannot be applied.
Therefore, to still establish some connection between the degrees of vertices of the sampled graphs and the original graph, we consider \cref{chernoff}, a generalized version of the Chernoff bound which holds for negatively correlated variables. 
Boolean random variables $X_1, X_2,\ldots, X_n$ are negatively correlated if any subset $S$ of $\{X_1, X_2, \ldots, X_n\}$ and any element $a\in S$ follows $\Pr(a=1|\forall_{b\in S-\{a\}}b=1)\leq \Pr(a = 1)$.
\begin{lemma}[Folklore]
\label{chernoff}
Let $X = \sum_{i=1}^n X_i$ where $X_1, X_2, \ldots X_n$ are negatively correlated boolean random variables. Then, for $\epsilon \in (0,1)$, we have that
\[\Pr(|X-\E{X}| > \epsilon \E{X}) \leq 3\exp(-\epsilon^2 \E{X}/3).\]
\end{lemma}

With \cref{chernoff}, we prove \cref{newgtoh}, which considers sampling using \SSSample, and \cref{newhtog}, which considers sampling a fixed amount of edges uniformly at random. These lemmas will be used to analyze \cref{alg:random}.

\begin{lemma}
\label{newgtoh}
Let $G = (V, E)$ be a graph, $\epsilon \in (0,1)$, and $p\in (0,1]$.
Let $H \gets \SSSample\rb{E', S, T, p, s}$ be obtained on \cref{line:streaming-H} of \cref{alg:random}.
Then, for all $v \in G$, the following hold with probability $1 - \frac{1}{n^3}$: (i) if $\din_G(v) \geq \xi / p$ then $p\din_G(v) \leq (1+\epsilon)\din_H(v)$, (ii) if $\din_G(v) < \xi / p$ then $\din_H(v) < 2 \xi$ (same claims hold for out-degrees $\dout_G(v)$ and $\dout_H(v)$).
\end{lemma}
\begin{proof}
Consider a sample $H$ of $G$ obtained by \SSSample. Number the edges in $E$ from $1$ to $m = |E|$. For each $i = 1 \ldots m$, define a Boolean random variable $X_i$ which equals $1$ iff $i \in H$ and $0$ otherwise.

Let $S$ be an arbitrary subset of $\{X_1, X_2, \ldots, X_m\}$ and $X_i$ be an arbitrary element of $S$.
Recall that \SSSample uses two sets to output $H$: one is the input set $E'$, and the other one is the rest of the stream $\Estream$.
To sample from $\Estream$, \SSSample first samples $x$ on \cref{line:sample-x}, and then fetches the next $x$ edges from $\Estream \cap E_G(S, T)$; hence, that edge-sample is correlated as the number of edges is fixed.
If $i \in E'$, then $X_i$ is independent of all the other random variables, and therefore, $\Pr(X_i=1|\forall_{b\in S-\{X_i\}}b=1)= \Pr(X_i = 1)$.
If $i \in \Estream$ instead, we observe that if the sample already contains an edge $j$, then it ``reduces'' the probability of $i$ appearing in $H$.
Overall, the random variables are negatively correlated.

Now, under the conditions of our lemma, we have that $p \leq 1$. So, when we use \SSSample, we see that
\begin{eqnarray*}
\E{\din_H(v)} &=& \E{\frac{p \cdot s}{|E_G(S,T)|}\cdot \din_G(v)}\\
&=& \frac{n\xi}{(1-\epsilon)|E_G(S,T)|}\din_G(v)
\end{eqnarray*}
since \SSSample, in expectation, samples $ps$ edges out of all the edges in $E_G(S,T)$. Notice that the result is the same as $\E{\din_H(v)}$ in \cref{gtoh}. Therefore, following the proof of \cref{gtoh} but using \cref{chernoff} instead of the Chernoff bound, we prove \cref{newgtoh}.
\end{proof}
\begin{lemma}
\label{newhtog}
Let $G = (V, E)$ be a graph, $\epsilon \in (0,1)$. Assuming $|E| \geq n\xi$, let $H$ be a sample of $n\xi$ edges from $G$ chosen uniformly at random with probability $p = \frac{n\xi}{|E|}$. Then, for all $v \in G$, the following hold with probability $1 - \frac{1}{n^3}$: (i) if $\din_H(v) \geq 2\xi$ then $|\din_H(v) - p\din_G(v)| \leq \epsilon \din_H(v)$, (ii) if $\din_H(v) < 2\xi$ then $\din_G(v) \leq \frac{2(1+\epsilon)\xi}{p}$ (same claims hold for out-degrees $\dout_G(v)$ and $\dout_H(v)$).
\end{lemma}
\begin{proof}
Note that we can apply \cref{chernoff} since we perform uniform random sampling with a fixed number of edges.
Therefore, when referencing \cref{gtoh} in this proof, it will be using \cref{chernoff} instead of the Chernoff bound in the context of uniform random sampling with a fixed number of edges.

Consider a vertex $v \in V$. Note that if $\din_H(v) \geq 2\xi$, then $\din_G(v) \geq \xi/p$ with probability $1 - \frac{1}{n^4}$ using the Chernoff bound. Therefore, with \cref{gtoh}, we have that $|\din_H(v) - p\din_G(v)| \leq \epsilon \din_H(v)$.

If $\din_H(v) < 2\xi$, we split our analysis into two cases. If $\din_G(v) \geq \xi/p$, then
\[\din_G(v) \leq \frac{(1+\epsilon)\din_H(v)}{p} < \frac{2(1+\epsilon)\xi}{p}\]
using \cref{gtoh}.
Otherwise, if $\din_G(v) < \xi/p$, then $\din_G(v) < \frac{2(1+\epsilon)\xi}{p}$ holds as well.
Therefore, the statements in \cref{newhtog} hold with probability $1 - \frac{1}{n^3}$ using the union bound.
The same proof applies to out-degrees.
\end{proof}

The following is our analysis of \cref{alg:random}.
\begin{lemma}
\label{newmem}
\cref{alg:random} uses $O\left(\frac{n\log^2 n}{\epsilon^3}\right)$ memory with probability $1 - \frac{\log_{1+\eps} n}{n^4}$.
\end{lemma}
\begin{proof}
First, note that our algorithm takes $O(\log_{1+\epsilon}n)$ iterations.

For the graph $H$, we either have $p \leq 1$ or $p > 1$. If $p \leq 1$, we create a sample of edges with probability $\frac{n \xi}{(1-\epsilon) s}$ using \SSSample.
Note that $s$ is an approximation of $|E_G(S,T)|$ where, with probability $1 - \frac{1}{n^4}$, it holds
\begin{equation}\label{eq:s-bound}
s \leq |E_G(S,T)| \leq \frac{1+\epsilon}{1-\epsilon}\cdot s
\end{equation}
using the Chernoff bound, \cref{lemma:chernoff}~\eqref{item:delta-at-most-1}.
So, for each of the $O(\log_{1+\epsilon}n)$ iterations, graph $H$ will have $O\left(\frac{(1+\epsilon)^2n\log n}{\epsilon^2(1-\epsilon)^2}\right) = O\left(\frac{n\log n}{\epsilon^2}\right)$ edges with probability $1 - \frac{1}{n^4}$ and $\epsilon < 0.9$.
On the other hand, if $p > 1$, we must have that
\[(1-\epsilon)s < n\xi \implies |E_G(S,T)| < \frac{(1+\epsilon)n\xi}{(1-\epsilon)^2}\]
with probability $1 - \frac{1}{n^4}$ using \eqref{eq:s-bound}. Then, in this case, graph $H$ will also have $O\left(\frac{n\log n}{\epsilon^2}\right)$ edges.

Note that $E'$ only increases by at most $n\xi$ edges each iteration in addition to edges from graph $H$. Therefore, using the union bound, the number of edges in $E'$ at any point in time is at most $O\left(\frac{n\log n \cdot \log_{1+\epsilon} n}{\epsilon^2}\right) = O\left(\frac{n\log^2 n}{\epsilon^3}\right)$ with probability $1 - \frac{\log_{1+\epsilon}n}{n^4}$.

Now, $E_A$ is either a uniform random sample of $n\xi$ edges from the remaining edges in the stream, since the stream is randomized, or is the rest of the stream.
If it is the rest of the stream, then our algorithm reduces to finding the densest subgraph in $E'\cup E_A(S,T)$. This uses at most $O\left(\frac{n\log^2 n}{\epsilon^3}\right)$ memory. However, if $E_A$ is not the rest of the stream, then either $|E_{A}(S,T)| < 2\xi$ or $|E_{A}(S,T)| \geq 2\xi$.

If $|E_{A}(S,T)| < 2\xi$, then we add the edges from $E_A(S,T)$ and $\Estream(S,T)$ to $E'$ and find the densest subgraph in $E'$. Using \cref{newhtog}, at most $\frac{2(1+\epsilon)\xi}{p} \leq 2(1+\epsilon)n$ edges, where $p = \frac{n\xi}{|\Estream(S,T)\cup E_{A}(S,T)|}$, are added to $E'$ from $E_A$ and the stream. This uses $O\left(\frac{n\log^2 n}{\epsilon^3}\right)$ memory.

If $|E_{A}(S,T)| \geq 2\xi$, then we have that the number of edges from $E_G(S,T)$ that were in the stream is at least $\frac{s(1-\epsilon)|E_A(S,T)|}{|E_A|}$ using \cref{newhtog}, allowing us to calculate $s'$ for creating graph $H$.\\

Therefore, the total amount of memory used by the algorithm is $O\left(\frac{n\log^2 n}{\epsilon^3}\right)$.
\end{proof}

\begin{theorem}\label{theorem:semi-streaming}
\cref{alg:random} produces a $2\left(1 +\epsilon\right)$-approximation of the densest directed subgraph with probability $1 - \frac{\log_{1+\epsilon}n}{n^2}$.
\end{theorem}
\begin{proof}
The pruning of vertices in our algorithm follows that of \cref{alg:first}, using \SSSample to create sample graph $H$ and then applying \VSetsUpdate. Therefore, following the proof of \cref{theorem:approx} but using \cref{newgtoh} instead of \cref{gtoh}, we attain the same approximation of the directed densest subgraph.
\end{proof}

\subsection{Removing the assumption on $\copt$ being known}
Executing \cref{alg:random} requires knowing $\copt$, the ratio between the two optimal sets. We remove this assumption by using standard techniques.
Namely, it is unclear how to learn $c$ without finding the densest subgraph.
To alleviate this, our algorithm guesses the value of $c$ and runs the algorithm described above for each guess.
Since there are $O(n^2)$ candidates for $c$, i.e., values $a/b$ for $a, b \in  \{1, \ldots, n\}$, our algorithm guesses the value of $c$ in the multiplicative increments of $\delta > 1$.
More precisely, we run our algorithm for values of $c$ ranging from $1/n$ to $n$ and of the form $\delta^i / n$.
As \cref{theorem:capprox} shows, this approach incurs an additional factor of $\sqrt{\delta}$ in our approximation. 
\begin{lemma}
\label{theorem:capprox}
The best density of \cref{alg:random} ran on values of $c = 1/n$ through $n$, multiplying by a constant factor of $\delta$, is a $2(1+\epsilon)\sqrt{\delta}$-approximation of the densest directed subgraph with probability $1-\frac{\log_{1+\epsilon}n}{n^2}$.
\end{lemma}
\begin{proof}
Consider $c$ where $\copt/\delta \leq c \leq \delta\copt$. Drawing from the proof of \cref{theorem:approx}, consider $\tS,\tT$ with $|\tS|/|\tT| = \copt$ that maximizes $\rho(\tS,\tT) = \rho^*(G)$. Then,
\begin{eqnarray}
\rho^*(G) &=& \frac{|E(\tS,\tT)|}{\sqrt{|\tS||\tT|}} \leq \frac{|\tS|\cdot d_{out}^* + |\tT|\cdot d_{in}^*}{\sqrt{|\tS||\tT|}}\nonumber\\
&=& \sqrt{\copt}\cdot d_{out}^* + \frac{1}{\sqrt{\copt}}\cdot d_{in}^*\nonumber\\
&\leq& \sqrt{\delta c}\cdot d_{out}^* + \sqrt{\delta/c}\cdot d_{in}^*\nonumber\\
&\leq& 2\left(1 +\epsilon\right)^3\sqrt{\delta}\cdot\rho(S^*,T^*)\nonumber
\end{eqnarray}
Therefore, this gives a $2(1+\epsilon)\sqrt{\delta}$-approximation of the densest directed subgraph.
\end{proof}
\section{MPC Algorithms}
In this section, we extend our approach developed in \cref{sec:single-pass} to two memory regimes of the MPC model.
We present the ideas gradually. First, we show that slightly adjusting \cref{alg:random} yields an $O(1)$ MPC round algorithm in the super-linear memory regime.
To achieve this round complexity, we prove a certain edge-size progress that was not required for the semi-streaming analysis.
Second, in \cref{sec:near-linear-MPC} we show that, by further building on the MPC super-linear memory algorithm, it is possible to find an approximate directed densest subgraph in only $O(\sqrt{\log n})$ rounds in the near-linear memory regime.

\subsection{Super-linear Memory Regime}
\label{sec:MPC-superlinear}
Recall that in \cref{alg:random}, $\Estream$ is the \emph{remaining} edges in the stream.
In developing a super-linear regime MPC algorithm, our main idea is to implement $\Estream$ efficiently in MPC. 
More concretely, assume that a machine $\cM$ aims to simulate \cref{alg:random}. Then, $\cM$ first samples $n^{1+\delta}$ edges from the graph -- let that set be $X_1$ -- and presents $X_1$ to \cref{alg:random} as a prefix of $\Estream$. 
Once \cref{alg:random} makes a pass over $X_1$, $\cM$ needs to fetch another set $X_2$ of $n^{1+\delta}$ edges from $E \setminus X_1$ and continue \cref{alg:random} with $X_2$. 
This process continues until there are edges to feed to \cref{alg:random}; the samples $X_1, X_2, \ldots$ represent $\Estream$.
Clearly, this simulation corresponds to \cref{alg:random} executed on a stream. Unfortunately, fetching sets $X_i$ can happen $m / n^{1+\delta}$ times, which for $m \gg n^{1+\delta}$ is highly inefficient.

Our main idea is that after processing $X_1$, it suffices if $\cM$ takes a random sample from $E_G(S, T) \cap (\Estream \setminus X_1)$ as opposed from $\Estream \setminus X_1$.
The main question here is: ``\emph{How does this translate to round complexity?}'' As the primary measure of progress, we show the following.
Let $(S', T')$ be the vertex set pair in \cref{alg:random} just before our simulation of \cref{alg:random} exhausted $X_1$.
Then, we show that, up to low-order terms, $|E_G(S', T')| < m / n^{\delta}$. More generally, we show that whenever our simulation needs to fetch fresh $X_i$ from the remaining unused edges, the number of \emph{relevant} remaining edges reduces by a factor of $n^{\delta}$.
Hence, after $O(1/\delta)$ many such steps, the number of remaining relevant edges fits in the memory of a single machine, at which point the rest of the problem can be solved locally.
We dive into the details of this argument in our proof of the following claim.




\begin{theorem}
\label{theorem:MPC-superlinear}
There exists an algorithm that
runs in $O(1/\delta)$ rounds with probability $1 - \frac{\log_{1+\epsilon}n}{n^4}$.
\end{theorem}
\begin{proof}
We first describe the algorithm and then an analysis of its round complexity.

\paragraph{Algorithm description.}
Let $\cAMPC$ be our MPC algorithm.
\begin{itemize}
    \item $\cAMPC$ maintains the set of \emph{relevant} edges, denoted by $\Erel$. Initially, $\Erel = E_G$.
    \item This algorithm simulates an instance of \cref{alg:random}, and only that instance. Let the instance be simulated on a machine $\cM$; the simulation proceeds in phases as follows:
    \begin{itemize}
        \item Let $(S_{i+1}, T_{i+1})$ be the vertex sets of \cref{alg:random} at the end of phase $i$ of $\cAMPC$. 
        \item At the beginning of phase $i + 1$, $\cAMPC$ updates $\Erel \gets \Erel \cap E_G(S_{i+1}, T_{i+1})$. 
        \item Then, $\cAMPC$ places on $\cM$ a set $X_{i+1}$ of $n^{1+ \delta}$ edges chosen uniformly at random from $\Erel$ and updates $\Erel \gets \Erel \setminus X_{i + 1}$.
        \item The rest of phase $i + 1$ consists of continuing the \cref{alg:random} instance with the set $X_{i+1}$.
    \end{itemize}
    \item $\cAMPC$ stops when $\Erel$ fits entirely on $\cM$.
\end{itemize}

All the steps of $\cAMPC$, except updating $\Erel$ and sampling $X_{i+1}$, are done on a single machine. 
Removing $X_i$ from $\Erel$ can be done by first sorting $\Erel$ and $X_i$ together and then removing that edge that has a duplicate in the sorted list. Sorting can be done in $O(1)$ rounds~\cite{goodrich2011sorting}.
Performing $\Erel \cap E_G(S_{i+1}, T_{i+1})$ can be done again by sorting $\Erel$, $S_{i+1}$ and $T_{i+1}$ altogether. Those edges incident to $S_{i + 1}$ or $T_{i + 1}$ are the only ones kept in $\Erel$.
To sample $X_{i+1}$, first, each edge in $\Erel$ samples a random integer in the range $[1, n^8]$. Whp, all the integers sampled by the edges are distinct. Second, all the edges are sorted with respect to the sampled integers; effectively, this creates a random permutation of $\Erel$. Third, the first $n^{1+\delta}$ edges in that sorted list are $X_{i+1}.$

\paragraph{Algorithm analysis.}
Consider a phase $i$ of $\cAMPC$ and let $(S_i, T_i)$ be our vertex sets before the pruning starts.
Now consider some pair of vertex sets $(S', T')$ while running \cref{alg:random} during phase $i$.
We define $\Erem$ as the union of $\Erel$ and the remaining edges in $X_i$ that the algorithm has not inspected so far.
There are $|E_G(S',T')| - |E'|$ edges in $E_G(S', T')$ remaining in $\Erem$. 

\cref{alg:random} needs at most $\frac{2n\xi}{1-\epsilon}$ edges in $E_G(S', T')$ from the stream in order to start peeling.
Let $k$ be the number of edges taken from the stream to perform one iteration of peeling. Therefore, using the Chernoff bound, \cref{lemma:chernoff}~\eqref{item:delta-at-most-1}, we see that
\[k \leq \frac{2n\xi |\Erem|}{(1-\epsilon)^2(|E_G(S',T')| - |E'|)}. \]
with probability $1 - \frac{1}{n^4}$, assuming that $|E_G(S',T')| - |E'|$ is at least $\Omega\left(n\log^2 n/\epsilon^2\right)$. If $|E_G(S',T')| - |E'|$ is $O\left(n\log^2 n/\epsilon^2\right)$, then all the edges in $\Erel$ for the next phase fit on $\cM$.

Now, if $k \geq \frac{n^{1+\delta}}{2\log_{1+\epsilon}n}$, then we have that
\begin{eqnarray*}
\frac{n^{1+\delta}}{2\log_{1+\epsilon}n} &\leq&  \frac{2n\xi |\Erem|}{(1-\epsilon)^2(|E_G(S',T')| - |E'|)}\\
\implies |E_G(S',T')| &\leq& \frac{4\xi\log_{1+\epsilon}n}{(1-\epsilon)^2 n^{\delta}}\cdot |\Erem| + |E'|.
\end{eqnarray*}
Recall from \cref{newmem} that $E'$ contains at most $O\left(\frac{n'\log^2 n}{\epsilon^3}\right)$ edges.
So, if $|\Erem| \leq |E'|$, then in the next phase $\Erel$ fits entirely on $\cM$.
On the other hand, if $|\Erem| > |E'|$, then whp $E_G(S',T')$ has $O\left(\frac{\log^2 n}{\epsilon^3 n^{\delta}}\cdot |\Erem| \right)$ edges for $\epsilon < 0.9$.
Therefore, since $\Erem$ is a subset of $E_G(S_i,T_i)$, this phase causes the number of edges between vertex sets $S$ and $T$ to reduce by a factor of $\Omega(n^{\delta})$.
Then, $\cAMPC$ takes $O(1/\delta)$ rounds if every phase has at least one pair of vertex sets with $k \geq \frac{n^{1+\delta}}{2\log_{1+\epsilon}n}$.

Consider now that $k < \frac{n^{1+\delta}}{2\log_{1+\epsilon}n}$ for all pairs of vertex sets in a phase. Observe that \cref{alg:random} takes at most $2\log_{1+\epsilon} n$ iterations, meaning that the entire algorithm will be finished this phase entirely on $\cM$.

Combining these two cases with respect to $k$, $\cAMPC$ takes $O(1/\delta)$ rounds.
\end{proof}

\subsection{Near-linear Memory Regime}
\label{sec:near-linear-MPC}
In this section, we improve the efficiency of our MPC algorithm for the super-linear memory regime, enabling us to obtain a quadratically faster algorithm for the near-linear memory regime than currently known~\cite{bahmani2012densest,bahmani2014efficient}. Our approach stems from two ideas.
The first idea is an observation: instead of sampling $n\xi$ edges, our algorithm requires sampling only $(|S| + |T|)\cdot \xi$ edges. 
Hence, once $|S| + |T|$ become sufficiently small, then the memory per machine becomes significantly larger than our sample sizes, and we can use ideas similar to those from the proof of \cref{theorem:MPC-superlinear} to improve the $O(\log n)$ round complexity. 
Unfortunately, this alone is not sufficient. To see that, consider an example in which $|S| / |T| = \Theta(n)$, e.g., the densest subgraph has a star-like structure.
In that case, it is potentially needed to peel $T$ for $O(\log n)$ times.
Since $|S| + |T| = \Theta(n)$ during those peeling steps, our observation does not yield any improvement.

The second idea is that as long as our algorithm peels only $T$ or only $S$, no fresh sampling is required. To elaborate, assume that the peeling process produces sets $(S, T_1), (S, T_2), \ldots, (S, T_j)$. 
Since $S$ remains the same, the estimated degree between $S$ and a vertex $v \in T_j$ is the same in all the sets $T_1, T_2, \ldots, T_j$. Hence, once $|E_G(S, v)|$ is estimated, no new sample is needed until $S$ changes. In other words, the star-like example we pointed to above can be handled with a single sample. \emph{Nevertheless, what does happen when $S$ changes?} In that case, the algorithm might need a new round to obtain a fresh sample. 
Fortunately, the size of $S$ reduces by at least a factor of $(1+\eps)$, increasing the gap between the memory per machine and $|S| + |T|$, which our algorithm utilizes in a way similar to described in \cref{sec:MPC-superlinear}. We formalize these ideas in the proof of the following claim.

\begin{theorem}
\label{theorem:MPC-nearlylinear}
There exists an algorithm that
runs in $O(\sqrt{\log_{1+\epsilon} n})$ rounds with probability $1 - \frac{\log_{1+\epsilon}n}{n^4}$.
\end{theorem}
\begin{proof}
Let $\cAMPC$ be the algorithm described in the proof of \cref{theorem:MPC-superlinear}. We extend $\cAMPC$ this algorithm to obtain the desired round complexity. In this proof, we only describe the extension.
\paragraph{Algorithm description.}
In a phase $i$, we calculate $E_G(u,T_i)$ and $E_G(S_i,v)$ for all $u\in S_i, v\in T_i$.
Then, if $|S_i|/|T_i| \geq c$, we peel from set $T$ using \VSetsUpdate until the resulting sets $(S_i', T_i')$ satisfy $|S_i'|/|T_i'| < c$.
Similarly, if $|S_i|/|T_i| < c$, we peel from set $S$ using \VSetsUpdate until the resulting sets $(S_i', T_i')$ satisfy $|S_i'|/|T_i'| \geq c$.
Finally, we do what $\cAMPC$ does every phase except that the edge samples now have size $(|S| + |T|)\cdot \xi$.

\paragraph{Algorithm analysis.}
Consider a phase $i$ of our algorithm and let $(S_i, T_i)$ be our vertex sets before the peeling starts.
Without loss of generality, assume that $|S_i|/|T_i| \geq c$.
Note that calculating the degrees of vertices and peeling from set $T$ can all be done in a constant number of rounds.
Afterwards, we have that the resulting sets $(S_i', T_i')$ satisfy $|S_i'|/|T_i'| < c$. 
This guarantees that both vertex sets are peeled from.

Now, let $n' = |S_i'| + |T_i'|$. We consider some pair of vertex sets $(S', T')$ while running \cref{alg:random} during phase $i$. With samples of $n'\xi$ edges, we follow a similar proof of \cref{theorem:MPC-superlinear} to see that if $k \geq \frac{n \poly(\log(n) / \eps)}{2\log_{1+\epsilon}n}$, either $\Erel$ fits entirely on $\cM$ in the next phase or $E_G(S', T')$ has $O\left(\frac{n'\log^2 n}{\epsilon^3 n \poly(\log(n) / \eps)}\cdot |\Erem| \right)$ edges. With the latter, the number of edges between vertex sets $S$ and $T$ reduces by a factor of $x = \Omega\left(\frac{\epsilon^3 n \poly \log n}{n'\log^2 n}\right)$.
Since both vertex sets are peeled from, $n'$ reduces by at least a factor of $1+\epsilon$ every phase. Then, the algorithm takes $O(\sqrt{\log_{1+\epsilon}n})$ rounds if every phase has at least one pair of vertex sets with $k \geq \frac{n \poly(\log(n) / \eps)}{2\log_{1+\epsilon}n}$ because $x$ reduces by at least a factor of $1+\epsilon$ every phase as well due to $n'$.

With $k < \frac{n \poly(\log(n) / \eps)}{2\log_{1+\epsilon}n}$ for all pair of vertex sets in a phase, the entire algorithm can be finished in this phase entirely on $\cM$. So, the algorithm takes $O(\sqrt{\log_{1+\epsilon} n})$ rounds.
\end{proof}
\section{Experiments} \label{sec:experiment}
We now demonstrate the practical performance of \cref{alg:random}, comparing it to Algorithm $3$ from \cite{bahmani2012densest}.
\vspace{-5pt}
\subsection{Data}
We use $4$ data sets from the Stanford Large Network Dataset Collection (\cite{snapnets}): Slashdot, Berkeley-Stanford Web Graph, LiveJournal, and Twitter.
In addition, we generate two synthetic graphs, PrefAttach $1$ and $2$, using the preferential attachment scheme \cite{simon1955class}.
This scheme is popular in the scientific analysis of graphs as it can generate the power-law distribution~\cite{barabasi1999emergence,jacob2015spatial,wan2017fitting} and, as such, it models the vertex-degree distribution of many important graphs, including World Wide Web~\cite{kunegis2013preferential} and Wikipedia~\cite{capocci2006preferential,gandica2015wikipedia}.

\begin{table}[h!]
\centering
 \begin{tabular}{|c c c|} 
 \hline
 Graph & Nodes & Edges \\
 \hline
 Slashdot & 82,168 & 948,464 \\ 
 Berk-Stan Web & 685,230 & 7,600,595  \\
 LiveJournal & 4,847,571 & 68,993,773 \\
 PrefAttach 1 & 100,000 & 100,051,302 \\
 PrefAttach 2 & 1,000,000 & 999,999,375 \\
 Twitter & 41,652,230 & 1,468,364,884 \\
 \hline
 \end{tabular}
\caption{Datasets we use in the evaluation.}
\label{table:1}
\end{table}

\begin{figure*}[ht!]
\centering
\includegraphics[width=0.32\textwidth]{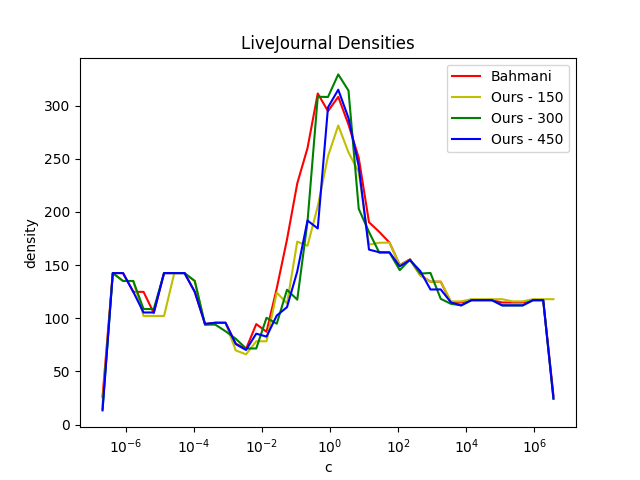}
\includegraphics[width=0.32\textwidth]{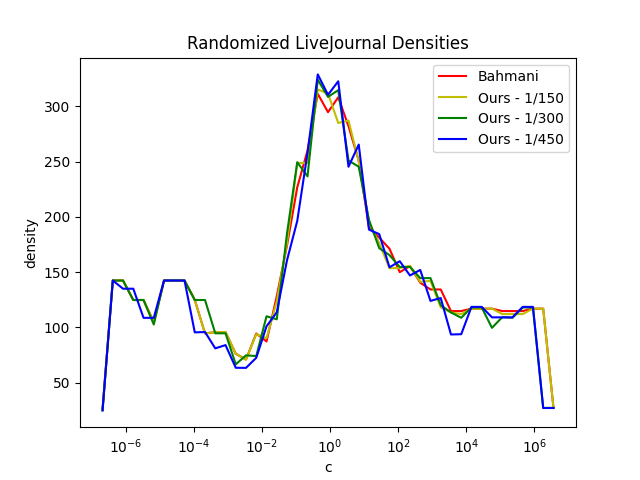}
\includegraphics[width=0.32\textwidth]{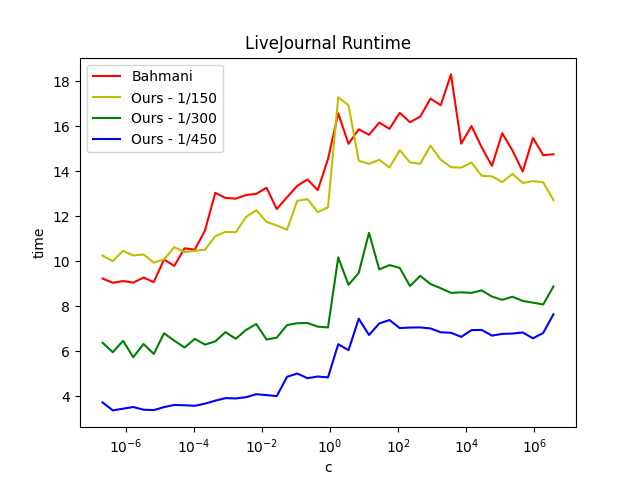}
\caption{Density and running-time as a function of $c$ for LiveJournal with $f = 1/150, 1/300, 1/450$.}
\label{fig:live}
\end{figure*}

\begin{figure*}[ht!]
\centering
\includegraphics[width=0.32\textwidth]{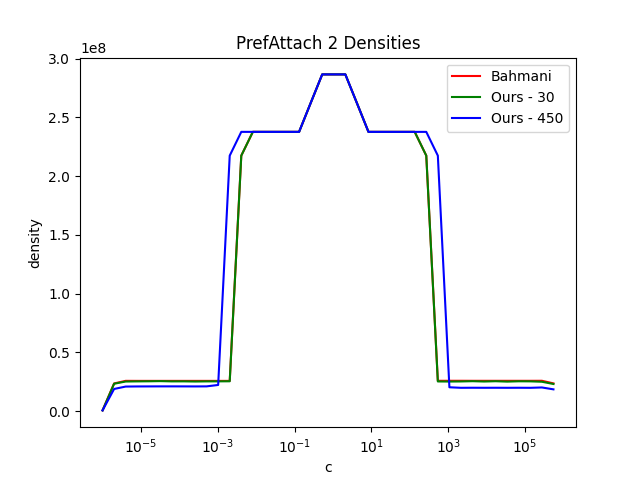}
\includegraphics[width=0.32\textwidth]{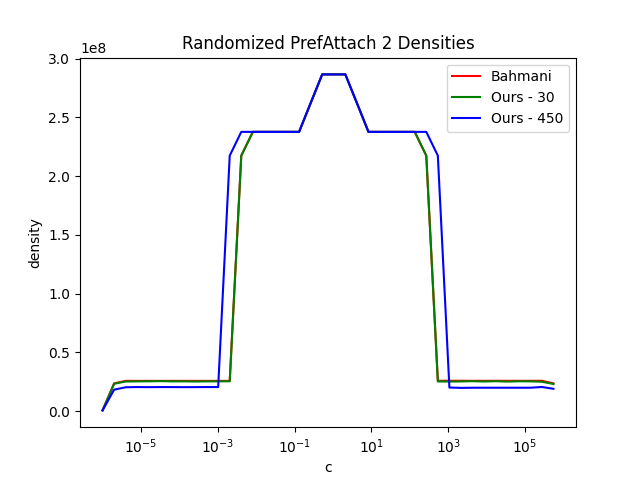}
\includegraphics[width=0.32\textwidth]{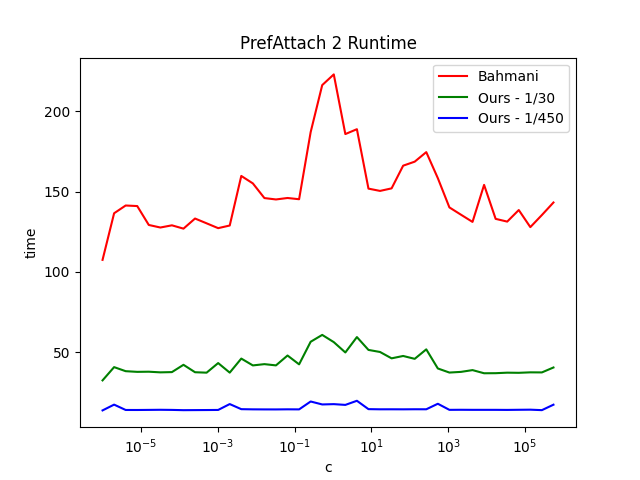}
\caption{Density and running-time as a function of $c$ for PrefAttach $2$ with $f = 1/30, 1/450$.}
\label{fig:graph2}
\end{figure*}

\begin{figure*}[ht!]
\centering
\includegraphics[width=0.32\textwidth]{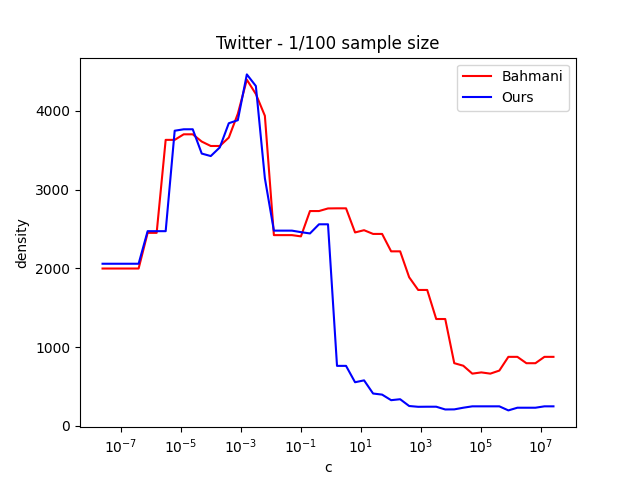}
\includegraphics[width=0.32\textwidth]{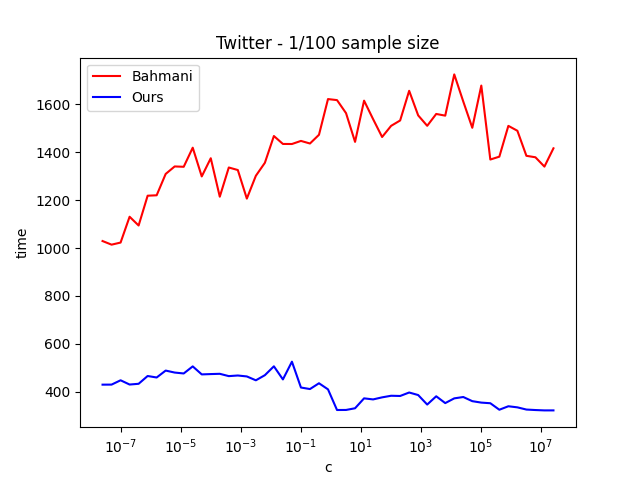}

\caption{Density and running-time as a function of $c$ for Twitter with $f = 1/100$.}
\label{fig:twitter}
\end{figure*}
\vspace{-15pt}
\subsection{Experimental Setup}
As a baseline, we use \cite{bahmani2012densest}.
Both algorithms, the one from \cite{bahmani2012densest} and ours, are implemented in C++ on a M1 machine running macOS Monterey 12.5.1 with 4 cores, 8 GB of RAM, 256 KB of L2 Cache, and 2.5 MB of L3 Cache (per core).
We run the algorithms for values of $c$ ranging from $1/n$ to $n$ with constant $\delta = 2$ and $\epsilon = 0.2$.
Then, we study the algorithms in terms of how close their approximation of the densest subgraph is and their runtime. Plots of density and runtime are created with respect to the constant $c$.

An important note about \cref{alg:random} is that it cannot calculate $\rho(S,T)$ for each iteration without seeing the whole graph.
However, when implementing the algorithm, we approximate that density through our sampled graph by multiplying the density in the sampled graph, $\rho_H(S,T)$, by $1/p$.

\newcommand{\f}{f}

\cref{alg:random} creates samples of size $n\xi = \frac{f \cdot n\log n}{\epsilon^2}$.
In the theoretical analysis, we set $\f = 60$.
In our evaluations, we also investigate the effect of $\f$ on the performance of our approach, expecting that $\f$ significantly smaller than $60$ would also result in accurate and fast execution of \cref{alg:random}.
To that end, in our implementation of \cref{alg:random}, we consider $\f \in \{1/10, 1/20, 1/30, 1/100, 1/150, 1/300, 1/450\}$.

Although our guarantees for \cref{alg:random} require randomized streams, to perform a faithful comparison to \cite{bahmani2012densest}, we also execute \cref{alg:random} \emph{\bf on arbitrarily ordered streams}.
In many cases, these arbitrarily ordered streams are edges sorted by their endpoints, i.e., these streams are far from being randomized.
\vspace{-5pt}
\subsection{Results}
Our algorithm finds subgraphs with densities very close to those output by our baseline Bahamani et al.~(VLDB 2021). For each graph, the densities with the best corresponding $f$ are within $3\%$ of each other.
Interestingly, this is the case even for relatively small values of $\f$, such as $\f = 1/450$, suggesting that relatively sparse samples can be leveraged in extracting dense subgraphs.
Moreover, our algorithm attains higher accuracy for large graphs than Bahamani et al.~(VLDB 2021).
For example, our algorithm's density is better by $1.5\%$ for Twitter (\cref{fig:twitter}) and by $5.5\%$ for LiveJournal (\cref{fig:live}).

In addition to the maximum density of our algorithm being similar to Bahamani et al.~(VLDB 2021), the shapes of the density plots with respect to $c$ are incredibly similar.
This further solidifies that, in terms of accuracy, \cref{alg:random} matches Bahamani et al.~(VLDB 2021) even on non-randomized streams.

In the cases of denser graphs, our algorithm is around $2$ times faster than the baseline. For PrefAttach $2$, it reaches around $4$ times faster (\cref{fig:graph2}) and for Twitter, around $2.5$ times faster (\cref{fig:twitter}). There is a balance between using smaller constants for our sample size, leading to the algorithm becoming faster, and the accuracy of the algorithm's density approximation. 

Note that the stream randomization does not affect our algorithm's performance, as illustrated in all the figures with randomized data.
In \cref{appendix:experiments}, we provide results of additional evaluations, including results for $\epsilon = 0.1$ which are similar to those described here.

\section{Conclusion and Future Work}
We studied the approximate directed densest subgraph problem.
We developed sampling strategies that enabled us to design efficient algorithms for finding dense subgraphs in semi-streaming and MPC.
Moreover, our approach appears to be highly practical -- we showed empirically that these sampling strategies yield faster algorithms even on non-randomized streams.

Our work was largely inspired by a discrepancy between state-of-the-art efficiency of finding undirected and directed dense subgraphs. Even though we made significant progress in bridging this gap, many interesting questions remain. We mention three that we find the most exciting.
(1) Is it possible to find a $\Theta(1)$-approximate directed densest subgraph in a non-randomized semi-streaming setting in $o(\log n)$ passes? As an intermediate question, studying the trade-off between the number of passes and the memory requirement would be interesting, e.g., interpolating between $1$ pass and $n^{3/2}$ memory, and $O(\log n)$ passes and $O(n \poly \log n)$ memory.
(2) For the same problem, does a $o(\log n)$ MPC round algorithm exist in the sublinear memory regime?
(3) Our MPC algorithm for the near-linear memory regime outputs a $2+\eps$ approximation. Is it possible to achieve a $1+\eps$ approximation with the same round complexity?


\section*{Acknowledgements}
We are grateful to Rishabh Bhaskaran for many insightful discussions on the empirical portion of our work.
S.~Mitrovi\' c was supported by the Google Research Scholar Program.

\bibliographystyle{alpha}
\bibliography{ref}

\clearpage

\appendix
\section{Additional empirical evaluation}
\label{appendix:experiments}
In this section, we provide empirical evaluation in addition to those given in the main paper.
We refer a reader to \cref{sec:experiment} for details on the baselines, datasets, and computational setup.

\begin{figure*}[ht!]
\includegraphics[width=0.32\textwidth]{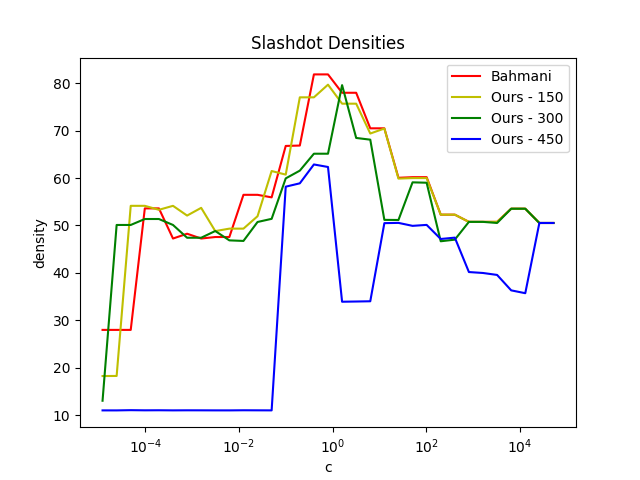}
\includegraphics[width=0.32\textwidth]{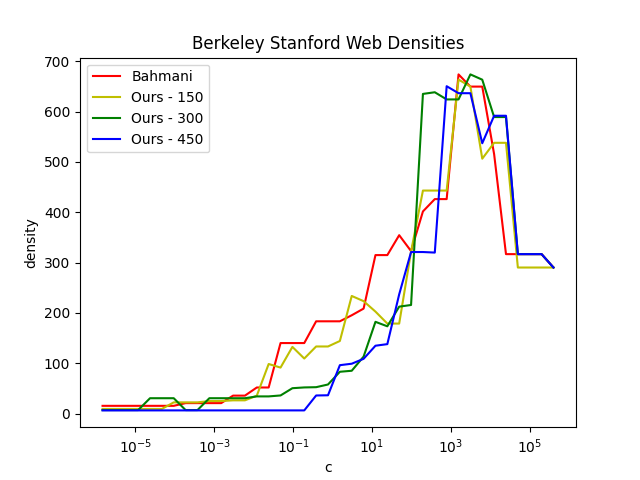}
\includegraphics[width=0.32\textwidth]{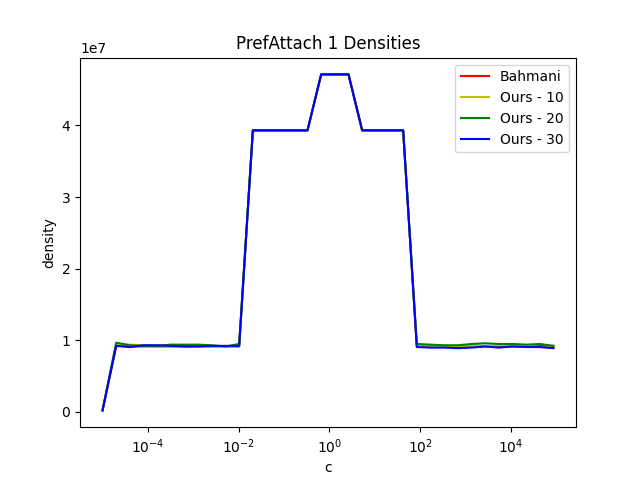}
\includegraphics[width=0.32\textwidth]{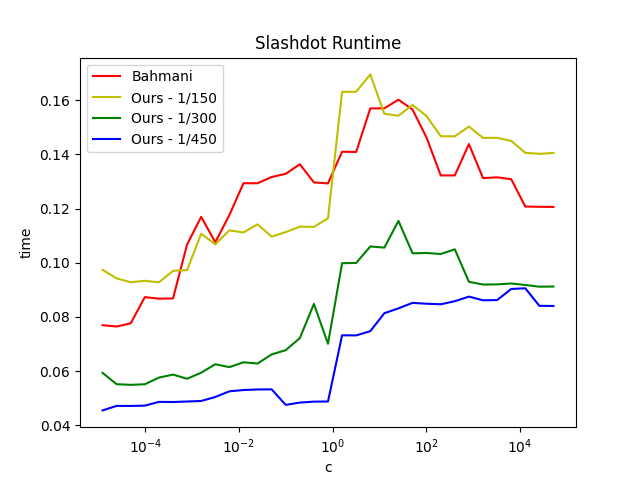}
\includegraphics[width=0.32\textwidth]{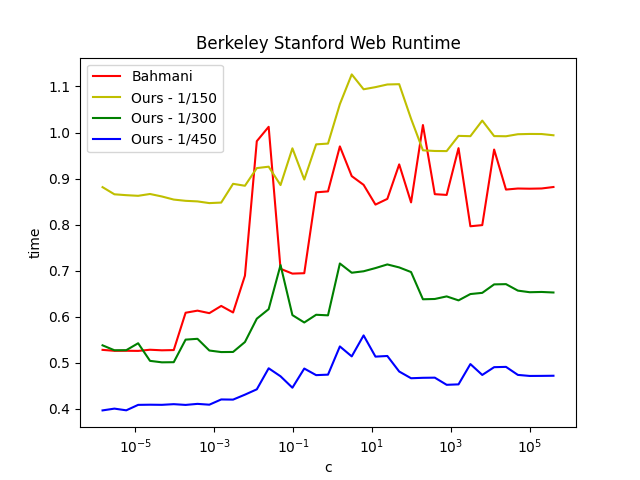}
\includegraphics[width=0.32\textwidth]{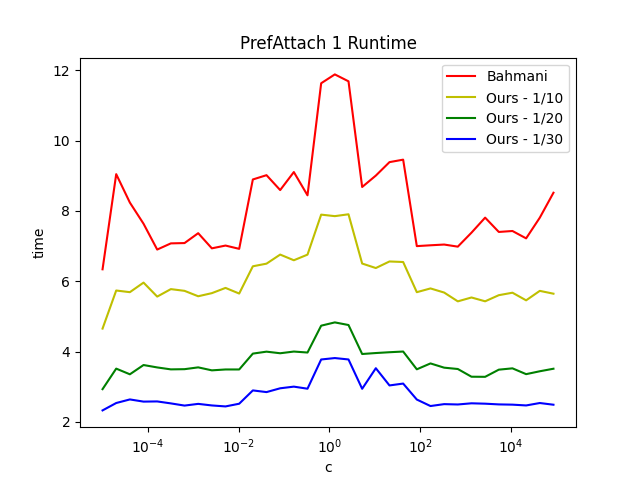}
\caption{Density and running-time as a function of $c$ for remaining graphs}
\label{fig:densityother}
\end{figure*}

\begin{figure*}[ht!]
\centering
\includegraphics[width=0.32\textwidth]{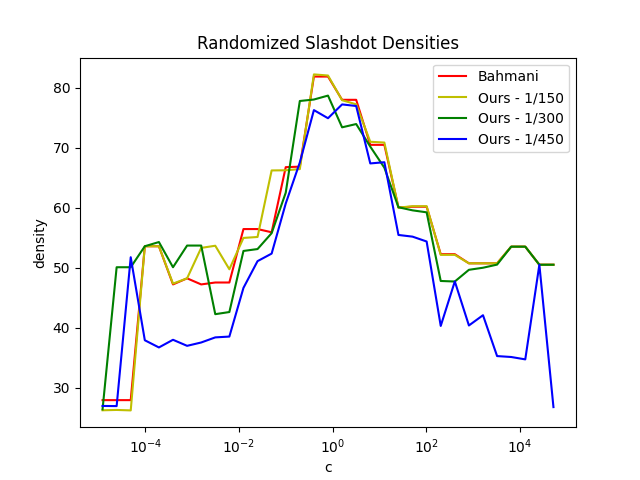}
\includegraphics[width=0.32\textwidth]{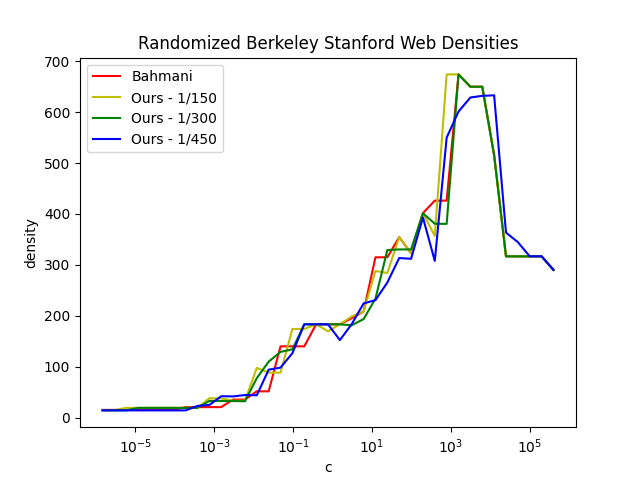}
\includegraphics[width=0.32\textwidth]{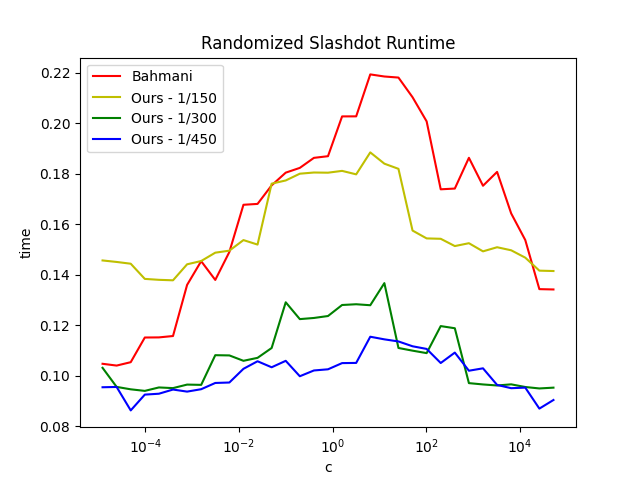}
\includegraphics[width=0.32\textwidth]{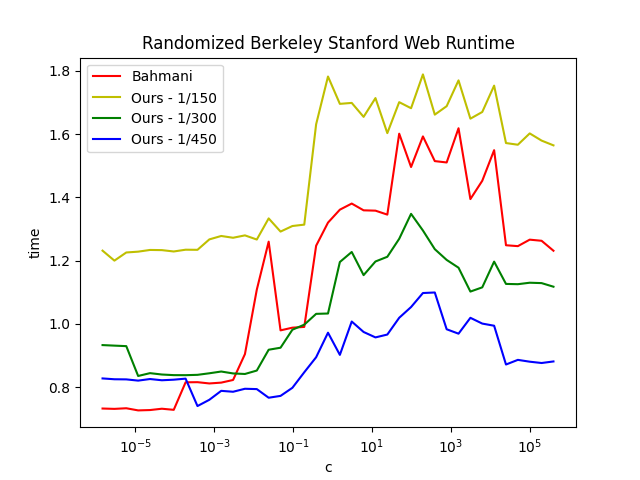}
\includegraphics[width=0.32\textwidth]{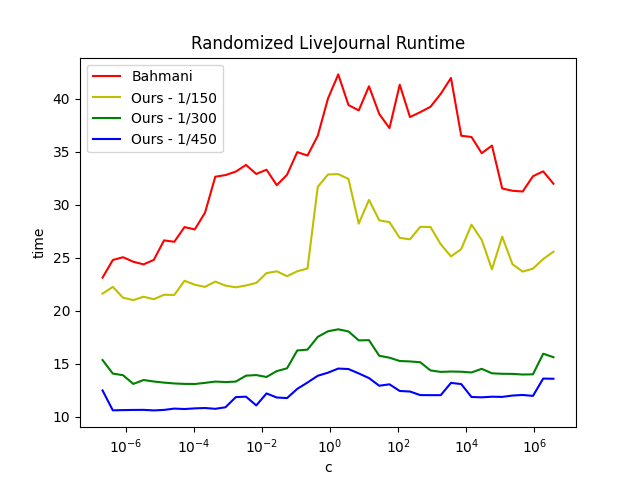}
\caption{Density and running-time as a function of $c$ for remaining graphs with randomization}
\label{fig:randomother}
\end{figure*}

\begin{figure*}[ht!]
\centering
\includegraphics[width=0.32\textwidth]{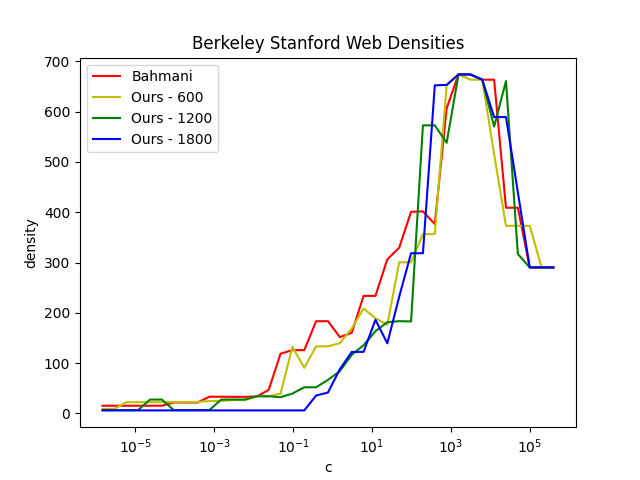}
\includegraphics[width=0.32\textwidth]{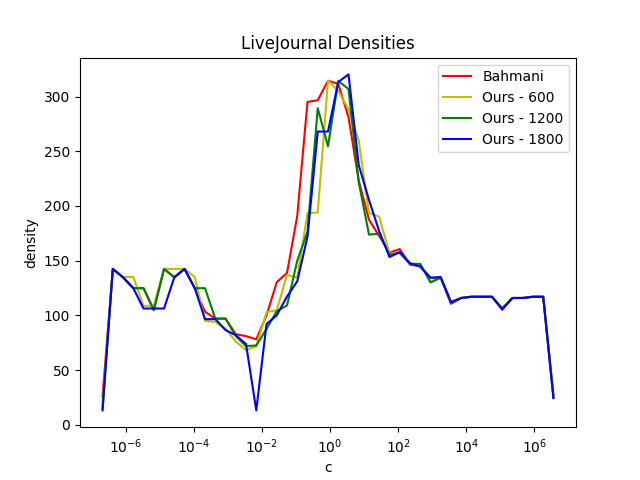}
\includegraphics[width=0.32\textwidth]{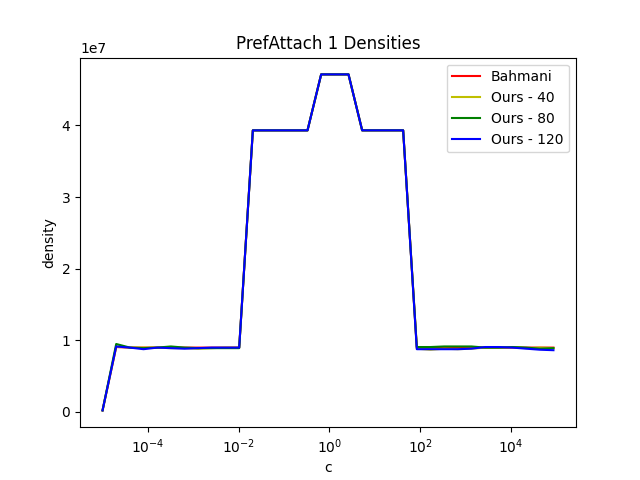}
\includegraphics[width=0.32\textwidth]{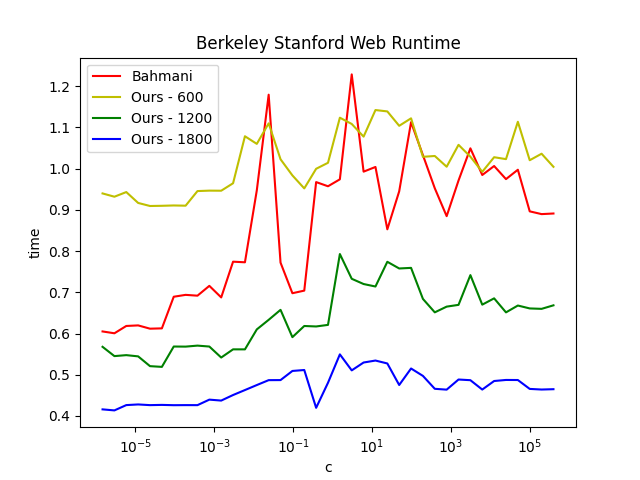}
\includegraphics[width=0.32\textwidth]{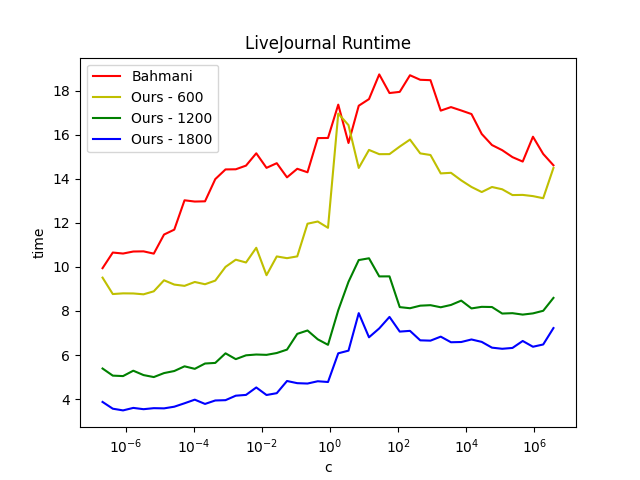}
\includegraphics[width=0.32\textwidth]{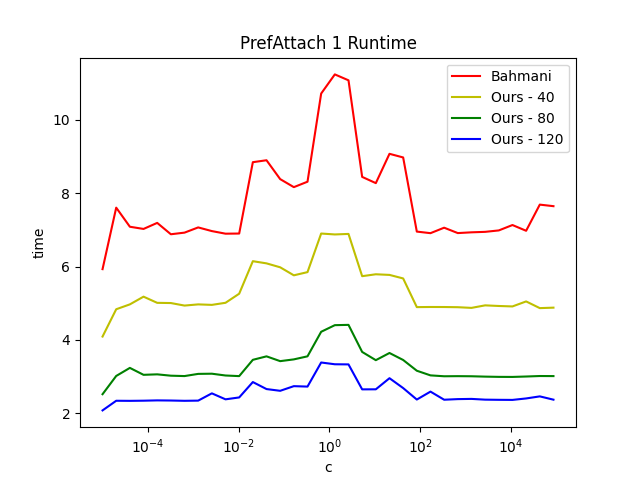}
\caption{Density and running-time as a function of $c$ for some graphs with $\epsilon = 0.1$}
\label{fig:epsilon0.1}
\end{figure*}

\end{document}